\documentclass[11pt,runningheads,a4paper]{llncs}

\usepackage{amssymb}
\usepackage{graphicx}
\usepackage{hyperref}
\hypersetup{
  colorlinks   = true,    
  urlcolor     = blue,    
  linkcolor    = blue,    
  citecolor    = red      
}

\usepackage{url}
\newcommand{\keywords}[1]{\par\addvspace\baselineskip
\noindent\keywordname\enspace\ignorespaces#1}

\usepackage{amsmath}

\begin{document}

\mainmatter  

\title{Efficient Computation of the Characteristic Polynomial of a Threshold Graph}

\titlerunning{Characteristic Polynomial of a Threshold Graph}

%
%
\author{Martin F\"urer%
\thanks{Research supported in part by NSF Grant CCF-1320814.}}
\authorrunning{Martin F\"urer}

\institute{Department of Computer Science and Engineering \\
	Pennsylvania State University \\
	University Park, PA 16802,  USA \\
	furer@cse.psu.edu \\
\url{http://www.cse.psu.edu/~furer}}

%
%

\maketitle

\begin{abstract}
An efficient algorithm is presented to compute the characteristic polynomial of a threshold graph. Threshold graphs were introduced by Chv\'atal and Hammer, as well as by Henderson and Zalcstein in 1977. A threshold graph is obtained from a one vertex graph by repeatedly adding either an isolated vertex or a dominating vertex, which is a vertex adjacent to all the other vertices. Threshold graphs are special kinds of cographs, which themselves are special kinds of graphs of clique-width 2. We obtain a running time of $O(n \log^2 n)$ for computing the characteristic polynomial, while the previously fastest algorithm ran in quadratic time.
\keywords{Efficient Algorithms, Threshold Graphs, Characteristic Polynomial}
\end{abstract}

\section{Introduction}

The characteristic polynomial of a graph $G=(V,E)$ is defined as the characteristic polynomial of its adjacency matrix $A$, i.e. $\chi(G, \lambda) 
 = \det(\lambda I - A)$. The characteristic polynomial is a graph invariant, i.e., it does not depend on the enumeration of the vertices of $G$. The complexity of computing the characteristic polynomial of a matrix is the same as that of matrix multiplication \cite{Keller-Gehrig85,Pernet2007} (see \cite[Chap.16]{BurgisserCS97}), currently $O(n^{2.376})$ \cite{CoppersmithW90}. For special classes of graphs, we expect to find faster algorithms for the characteristic polynomial. Indeed, for trees, a chain of improvements \cite{TinhoferS85,Mohar89} resulted in an $O(n \log^2 n)$ time algorithm \cite{Furer2014}. The determinant and rank of the adjacency matrix of a tree can even be computed in linear time \cite{FrickeHJT96}. For Threshold graphs (defined below), Jacobs et al.\ \cite{JacobsTT2014} have designed an $O(n^2)$ time algorithm to compute the characteristic polynomial. Here, we improve the running time to $O(n \log^2 n)$. As usual, we use the algebraic complexity measure, where every arithmetic operation counts as one step. Throughout this paper, $n=|V|$ is the number of vertices of $G$.

Threshold graphs \cite{ChvatalH77,HendersonZ77} are defined as follows. Given $n$ and a sequence $b=(b_1,\dots, b_{n-1}) \in \{0,1\}^{n-1}$, the threshold graph $G_b = (V, E)$ is defined by $V=[n]=\{1,\dots,n\}$, and for all $i<j$, $\{i,j\} \in E$ iff $b_i = 1$. Thus $G_b$ is constructed by an iterative process starting with the initially isolated vertex $n$. In step $j>1$, vertex $n-j+1$ is added. At this time, vertex $j$ is isolated if $b_j$ is 0, and vertex $j$ is adjacent to all other (already constructed) vertices $\{j+1,\dots,n\}$ if $b_j=1$. It follows immediately that $G_b$ is isomorphic to $G_{b'}$ iff $b=b'$. $G_b$ is connected if $b_1 = 1$, otherwise vertex $1$ is isolated. Usually, the order of the vertices being added is $1,2, \dots ,n$ instead of $n,n-1, \dots, 1$. We choose this unconventional order to simplify our main algorithm.

Threshold graphs have been widely studied and have several applications from combinatorics to computer science and psychology \cite{MahadevP95}.

In the next section, we study determinants of weighted threshold graph matrices, a class of matrices containing adjacency matrices of threshold graphs. In Section 3, we design the efficient algorithm to compute the characteristic polynomial of threshold graphs. We also look at its bit complexity in Section 4, and finish with open problems.

\section{The determinant of a weighted threshold graph matrix}

We are concerned with adjacency matrices of threshold graphs, but we consider a slightly more general class of matrices. We call them weighted threshold graph matrices.
Let $M_{b_1 b_2 \dots b_{n-1}}^{d_1 d_2 \dots d_n}$ be the matrix with the following entries.
\[ {\left(M_{b_1 b_2 \dots b_{n-1}}^{d_1 d_2 \dots d_n}\right)}_{ij} =
\begin{cases}
b_i & \mbox{if $i<j$} \\
b_j & \mbox{if $j<i$} \\
d_i & \mbox{if $i=j$}
\end{cases}
 \]
 Thus, the weighted threshold matrix for $(b_1 b_2 \dots b_{n-1};d_1 d_2 \dots d_n)$ looks like this.
 \[M_{b_1 b_2 \dots b_{n-1}}^{d_1 d_2 \dots d_n} =
\begin{pmatrix}
 d_1 & b_1 & b_1 & \dots & b_1 & b_1 \\
 b_1 & d_2 & b_2 & \dots & b_2 & b_2 \\
 b_1 & b_2 & d_3 & \dots & b_3 & b_3 \\
 \vdots & \vdots & \vdots & \ddots & \vdots & \vdots \\
 b_1 & b_2 & b_3 & \dots & d_{n-1} & b_{n-1} \\
 b_1 & b_2 & b_3 & \dots & b_{n-1} & d_n \\
\end{pmatrix}
 \]
 
 In order to compute the determinant of $M_{b_1 b_2 \dots b_{n-1}}^{d_1 d_2 \dots d_n}$, we subtract the penultimate row from the last row and the penultimate column from the last column. In other words, we do a similarity transform with the following regular matrix
 \[ P=
\begin{pmatrix}
 1 & 0  & 0 & \dots & 0 & 0 \\
 0 & 1 & 0 & \dots & 0 & 0 \\
0 & 0 & 1 & \dots & 0 & 0 \\
 \vdots & \vdots & \vdots & \ddots & \vdots & \vdots \\
0 & 0 & 0 & \dots & 1 & -1 \\
0 & 0 & 0 & \dots & 0 & 1 \\
\end{pmatrix},
 \]
i.e.,
 \[ P_{ij} =
\begin{cases}
1 & \mbox{if $i=j$} \\
-1 & \mbox{if $i=n$ and $j=n-1$} \\
0 & \mbox{otherwise}.
\end{cases}
 \]
 The row and column operations applied to $M_{b_1 b_2 \dots b_{n-1}}^{d_1 d_2 \dots d_n}$ produce the similar matrix
 \[P^T \, M_{b_1 b_2 \dots b_{n-1}}^{d_1 d_2 \dots d_n} \, P =
 \begin{pmatrix}
 d_1 & b_1 & b_1 & \dots & b_1 & 0 \\
 b_1 & d_2 & b_2 & \dots & b_2 & 0 \\
 b_1 & b_2 & d_3 & \dots & b_3 & 0 \\
 \vdots & \vdots & \vdots & \ddots & \vdots & \vdots \\
 b_1 & b_2 & b_3 & \dots & d_{n-1} & b_{n-1}-d_{n-1} \\
 0 & 0 & 0 & \dots & b_{n-1}-d_{n-1} \, & \, d_n + d_{n-1} - 2b_{n-1} \\
\end{pmatrix}
 \]
 Naturally, the determinant of $P$ is 1, implying 
 \[\det \left( P^T \, M_{b_1 b_2 \dots b_{n-1}}^{d_1 d_2 \dots d_n} \, P \right)
 = \det \left( M_{b_1 b_2 \dots b_{n-1}}^{d_1 d_2 \dots d_n} \right). 
 \]
 Furthermore, we observe that $P^T \, M_{b_1 b_2 \dots b_{n-1}}^{d_1 d_2 \dots d_n} \, P$ has a very nice pattern.
 \[ P^T \, M_{b_1 b_2 \dots b_{n-1}}^{d_1 d_2 \dots d_n} \, P =
 \left( \,
 \begin{array}{cccccc}
 \cline{1-5}
\multicolumn{5}{|c|}{M_{b_1 b_2 \dots b_{n-2}}^{d_1 d_2 \dots d_{n-1}}} & 
\begin{array}{c}
b_1 \\ b_2 \\  b_3 \\ \vdots \\ b_{n-1}-d_{n-1}
\end{array} \\
 \cline{1-5} 
 \rule{0mm}{4mm} b_1 & b_2 & b_3 & \dots & b_{n-1}-d_{n-1} \, & \, d_n + d_{n-1} - 2b_{n-1} \\
\end{array}
\right)
 \]
To further compute the determinant of $P^T \, M_{b_1 b_2 \dots b_{n-1}}^{d_1 d_2 \dots d_n} \, P$,
we use Laplacian expansion by minors applied to the last row. 
\begin{eqnarray*}
\lefteqn{\det \left( M_{b_1 b_2 \dots b_{n-1}}^{d_1 d_2 \dots d_n} \right)
	=  \det \left(P^T M_{b_1 b_2 \dots b_{n-1}}^{d_1 d_2 \dots d_n} P \right)} \\
	& = & (d_n + d_{n-1} - 2 b_{n-1}) \det \left( M_{b_1 b_2 \dots b_{n-2}}^{d_1 d_2 \dots d_{n-1}} \right)
		- (b_{n-1} - d_{n-1})^2 \det \left( M_{b_1 b_2 \dots b_{n-3}}^{d_1 d_2 \dots d_{n-2}} \right) \\
\end{eqnarray*}
By defining the determinant of the $0 \times 0$ matrix $M_{b_1 b_2 \dots b_{n-1}}^{d_1 d_2 \dots d_n}$ with $n=0$ to be 1, and checking the determinants for $n=1$ and $n=2$ directly, we obtain the following result.
\begin{theorem} \label{thm:rec}
 $D_n = \det \left( M_{b_1 b_2 \dots b_{n-1}}^{d_1 d_2 \dots d_n}\right) $ is determined by the recurrence equation
 \[D_n =
 \begin{cases}
1 & \mbox{if $n=0$} \\
d_1 & \mbox{if $n=1$} \\
(d_n + d_{n-1} - 2 b_{n-1}) D_{n-1} - (b_{n-1} - d_{n-1})^2 D_{n-2} & \mbox{if $n \geq 2$}
\end{cases} 
\]\qed
\end{theorem}
This has an immediate implication, as we assume every arithmetic operation takes only 1 step.
\begin{corollary}
 The determinant of an $n \times n$ weighted threshold graph matrix can be computed in time $O(n)$.
\end{corollary}
\begin{proof}
 Every step of the recurrence takes a constant number of arithmetic operations. 
\end{proof} \qed

For arbitrary matrices, the tasks of computing matrix products, matrix inverses, and determinants are all equivalent \cite[Chap.16]{BurgisserCS97}, currently $O(n^{2.376})$ \cite{CoppersmithW90}. For weighted threshold graph matrices, they all seem to be different. We have just seen that the determinant can be computed in linear time, which is optimal, as this time is already needed to read the input. The same lower bound holds for computing the characteristic polynomial, and we will show an $O(n \log^2 n)$ algorithm. It is not hard to see that the multiplication of weighted threshold graph matrices can be done in quadratic time. This is again optimal, because the product is no longer a threshold graph matrix, and its output requires quadratic time.

\section{Computation of the Characteristic Polynomial of a Threshold Graph}
 
 The adjacency matrix $A$ of the $n$-vertex threshold graph $G$ defined by the sequence $(b_1, \dots , b_{n-1})$ is the matrix 
 $M_{b_1 b_2 \dots b_{n-1}}^{0 \, 0 \dots 0}$, and the characteristic polynomial of this threshold graph is
 \[ \chi(G,\lambda) = \det(\lambda I - A) = \det \left( M_{-b_1 -b_2 \dots -b_{n-1}}^{\lambda \, \lambda \dots \lambda}\right). \]
 This immediately implies that any value of the characteristic polynomial can be computed in linear time.
 
 The characteristic polynomial itself can be computed by the recurrence equation of Theorem~\ref{thm:rec}. Here all $d_i = \lambda$, and $D_n$, as the characteristic polynomial of an $n$-vertex graph, obviously is a polynomial of degree $n$ in $\lambda$.  Now, the computation of $D_n$ from $D_{n-1}$ and $D_{n-2}$ according to the recurrence equation is a multiplication of polynomials. It takes time $O(n)$, as one factor is always of constant degree. The resulting total time is quadratic. The same quadratic time is achieved, when we compute the characteristic polynomial $\chi(G,\lambda)$ for $n$ different values of $\lambda$ and interpolate to obtain the polynomial $\chi(G,\lambda)$.
 
 We want to do better. Therefore, we write the recurrence equation of Theorem~\ref{thm:rec} in matrix form.
 \[
\begin{pmatrix}
 D_n \\
 D_{n-1} \\
\end{pmatrix}
=
\begin{pmatrix}
 d_n + d_{n-1} - 2 b_{n-1} \,\, & \,\, -(b_{n-1} - d_{n-1})^2 \\
 1 & 0 \\
\end{pmatrix}
 \begin{pmatrix}
 D_{n-1} \\
 D_{n-2} \\
\end{pmatrix}
 \]
 Noticing that $D_0 =1$ and $D_1 = \lambda$, and all $d_i = \lambda$, we obtain the following matrix recurrence immediately.
\begin{theorem} \label{thm:Dn}
 For 
 \[ B_n = \begin{pmatrix}
 2(\lambda - b_{n}) \,\, & \,\, -(b_{n} - \lambda)^2 \\
 1 & 0 \\
\end{pmatrix},
\]
we have 
 \[
\begin{pmatrix}
 D_n \\
 D_{n-1} \\
\end{pmatrix}
=
B_{n-1} B_{n-2} \cdots B_1
 \begin{pmatrix}
\lambda \\
1 \\
\end{pmatrix}
 \] \qed
\end{theorem} 
This results in a much faster way to compute the characteristic polynomial $\chi(G,\lambda)$.
\begin{corollary}
The characteristic polynomial $\chi(G,\lambda)$ of a threshold graph $G$ with $n$ vertices can be computed in time $O(n \log^2 n)$.
\end{corollary}
\begin{proof}
 For  every $i$, all the entries in the $2\times 2$ matrix $B_i$ are polynomials in $\lambda$ of degree at most 2. Therefore, products of any $k$ such factors have entries which are polynomials of degree at most $2k$. To be more precise, actually the degree bound is $k$, because by induction on $k$, one can easily see that the degrees of the $i,j$-entry of such a matrix is at most
 \begin{eqnarray*}
k & \mbox{for $i=1$ and $j=1$,} \\
k+1 & \mbox{for $i=1$ and $j=2$,} \\
k-1 & \mbox{for $i=2$ and $j=1$,} \\
k & \mbox{for $i=2$ and $j=2$,} \\
\end{eqnarray*}
But the bound of $2k$ is sufficient for our purposes. W.l.o.g., we may assume that $n-1$ (the number of factors) is a power of 2. Otherwise, we could fill up with unit matrices.
Now the product $B_{n-1} B_{n-2} \cdots B_1$ is computed in $\log (n-1)$ rounds of pairwise multiplication to reduce the number of factors by two each time. In the $r$th round, we have $n 2^{-r}$ pairs of matrices with entries of degree at most $2^r$, requiring $O(n2^{-r})$ multiplications of polynomials of degree at most $2^r$. With FFT (Fast Fourier Transform) this can be done in time $O(n r)$. Summing over all rounds $r$ results in a running time of $O(n \log^2 n)$.
\end{proof} \qed

Omitting the simplification of $d_i = \lambda$  in Theorem \ref{thm:Dn}, we see immediately, that also the characteristic polynomial of a weighted threshold graph matrix can be computed in the same asymptotic time of $O(n \log^2 n)$.

\section{Complexity in the Bit Model}
By definition, the characteristic polynomial of an $n$-vertex graph can be viewed as a sum of $n!$ monomials with coefficients form $\{-1,0,1\}$. Thus all coefficients of the characteristic polynomial have absolute value at most $n!$, and can therefore be represented by binary numbers of length $O(n \log n)$. The coefficients can indeed be so big. An example is the constant term in the characteristic polynomial of the clique $K_n$. Its absolute value is the number of derangements (permutations without fixed points), which asymptotically converges to $n!/e$.

With such long coefficients, the usual assumption of arithmetic operations in linear time is actually unrealistic for large $n$. Therefore, the bit model might be more useful. We can use the Turing machine time, because our algorithm is sufficiently uniform. No Boolean circuit is known to compute such things with asymptotically fewer operations than the number of steps of a Turing machine.

We can use the fast $m \log m 2^{O(\log^* m)}$ integer multiplication algorithm~\cite{Furer09} (where $m$ is the length of the factors) to compute 
the FFT for the polynomials. A direct implementation, just using fast integer multiplication everywhere, results in time 
\[(n r) (2^r r^2  2^{O(\log^* r)}) = 
n 2^r r^3 2^{O(\log^* r)}\] 
for the $r$th round where $O(n 2^{-r})$ pairs of polynomials of degree $O(2^r)$ are multiplied. The coefficients of these polynomials have length $O(2^r r)$.
As the coefficients and the degrees of the polynomials increase at least geometrically, only the last round with $r = \log n$ counts asymptotically. The resulting time bound is $O(n^2 \log^3 n 2^{O(\log^* n)})$.
Using Sch\"onhage's \cite{Schonhage1982} idea of encoding numerical polynomials into integers in order to do polynomial multiplication, a speed-up is possible. Again only the last round matters. Here a constant number of polynomials of degree $O(n)$ with coefficients of length $O(n \log n)$ are multiplied. For this purpose, each polynomial is encoded into a number of length $O(n^2 \log n)$, resulting in a computation time of 
\[n^2 \log^2 n \, 2^{O(\log^* n)}.\]

Actually, because the lengths of coefficients are not smaller than the degree of the polynomials, no encoding of polynomials into numbers is required for this speed-up. In this case, one can do the polynomial multiplication in a polynomial ring over Fermat numbers as in Sch\"onhage and Strassen \cite{SchonhageS1971}. Then, during the Fourier transforms all multiplications are just shifts. Fast integer multiplication is only used for the multiplication of values. This results in the same asymptotic $n^2 \log^2 n \, 2^{O(\log^* n)}$ computation time with a better constant factor.

\section{Open Problems}
We have improved the time to compute the characteristic polynomial of a threshold graph from quadratic to almost linear (in the algebraic model). The question remains whether another factor of $\log n$ can be removed. More interesting is the question whether similarly efficient algorithms are possible for richer classes of graphs. Of particular interest are larger classes of graphs containing the threshold graphs, like cographs, graphs of clique-width 2, graphs of bounded clique-width, or even perfect graphs.

\end{document}